\documentclass[twoside]{article}

\usepackage[accepted]{aistats2020}
\usepackage[round]{natbib}

\usepackage{lstbayes}
\usepackage{listings}

\usepackage{amsmath}
\usepackage{amssymb}
\usepackage{amsthm}
\usepackage{url}            
\usepackage{graphicx}
\usepackage{booktabs}       
\usepackage{amsfonts}       
\usepackage{microtype}
\usepackage{enumitem}

\usepackage{amsfonts}       
\usepackage{nicefrac}       
\usepackage{microtype}      
\usepackage{scalerel}       

\setlist{nosep}

\usepackage{hyperref}       
\hypersetup{unicode=true,
           pdfkeywords={keywords},
           pdfborder={0 0 0},
           breaklinks=true,
           colorlinks=true,
           linkcolor=black,
           citecolor=black,
           filecolor=black,
           urlcolor=black,
         }

\newcommand{\E}{{\mathbb E}}

\newcommand{\Cov}{{\mathbb C}{\rm ov}}
\newcommand{\RR}{{\mathbb R}}

\newcommand{\n}{\|}
\newcommand{\bi}{\begin{itemize}}
\newcommand{\ei}{\end{itemize}}
\newcommand{\be}{\begin{enumerate}}
\newcommand{\ee}{\end{enumerate}}
\newcommand{\ra}{\rightarrow}

\newcommand{\iy}{\infty}

\newcommand{\beq}{\begin{equation}}
\newcommand{\eeq}{\end{equation}}
\newcommand{\beqa}{\begin{eqnarray*}}
\newcommand{\eeqa}{\end{eqnarray*}}
\newcommand{\btm}{\begin{theorem}}
\newcommand{\etm}{\end{theorem}}
\newcommand{\bpf}{\begin{proof}}
\newcommand{\epf}{\end{proof}}
\newcommand{\bla}{\begin{lemma}}
\newcommand{\ela}{\end{lemma}}
\newcommand{\bdn}{\begin{definition}}
\newcommand{\edn}{\end{definition}}
\newcommand{\bpn}{\begin{proposition}}
\newcommand{\epn}{\end{proposition}}
\newcommand{\bcy}{\begin{corollary}}
\newcommand{\ecy}{\end{corollary}}
\DeclareMathOperator{\tr}{tr}
\DeclareMathOperator*{\foo}{\scalerel*{+}{\sum}}

\newtheorem{proposition}{Proposition}
\newtheorem{definition}{Definition}
\newtheorem{lemma}[proposition]{Lemma}
\newtheorem{corollary}[proposition]{Corollary}

\begin{document}

%

%

\twocolumn[

\aistatstitle{Leave-One-Out Cross-Validation for Bayesian Model Comparison in Large Data}

\aistatsauthor{ M\r{a}ns Magnusson \And Michael Riis Andersen \And  Johan Jonasson \And  Aki Vehtari }


\aistatsaddress{ Aalto University \And Technical University\\ of Denmark \And Chalmers University\\ of Technology \And Aalto University} ]

\begin{abstract}
Recently, new methods for model assessment, based on subsampling and posterior approximations, have been proposed for scaling leave-one-out cross-validation (LOO) to large datasets. Although these methods work well for estimating predictive performance for individual models, they are less powerful in model comparison. We propose an efficient method for estimating differences in predictive performance by combining fast approximate LOO surrogates with exact LOO subsampling using the difference estimator and supply proofs with regards to scaling characteristics. The resulting approach can be orders of magnitude more efficient than previous approaches, as well as being better suited to model comparison.
\end{abstract}

\section{INTRODUCTION}
Model comparison is an important part of probabilistic machine learning. In many real-world domains, we are often confronted with multiple models and would like to choose the model that best generalizes to new, unseen data. This can be done in a large number of ways, but here we will restrict ourselves to choosing between models based on the \emph{predictive} performance. Due to the growing data sizes over the last years, scaling model comparison methods to large data is an important problem.

One measure of predictive performance is the \emph{expected log predictive density} (elpd) given by
\begin{align}
\label{elpd}
\overline{\text{elpd}}_M= &\int \log p_M(\tilde{y}_i|y) p_t(\tilde{y}_i) d\tilde{y}_i \\ = & \int \log \left[ \int p_M(\tilde{y}_i|\theta) p_M(\theta|y) d\theta \right] p_t(\tilde{y}_i) d\tilde{y}_i \nonumber ,
\end{align}
where $\log p_M(\tilde{y}_i|y)$ is the log predictive density of model $M$ for a new observation $\tilde{y}_i$, that has been generated by some true, unknown process, $p_t(\tilde{y}_i)$. The log predictive density, or the log score, has good theoretical properties in that it is both \emph{local}, i.e., only depend on $\tilde{y}_i$, and \emph{proper}, the  expected reward is maximized by the true probability distribution \citep{bernardo1979expected,bernardo1994bayesian,gneiting2007probabilistic,vehtari2012survey}. Although we focus on the log score in this paper, other scoring functions can be used.


\subsection{Leave-one-out cross-validation}

Leave-one-out cross-validation (LOO) is a method for estimating the elpd, or the generalization performance, of a model \citep{bernardo1994bayesian,vehtari2012survey,vehtari2017practical}. This is done by training the model on all observations except observation $y_i$, and then predicting the hold-out observation $y_i$, something that is then repeated for all $n$ observations. In this way we treat each observation $y_i$ as a pseudo-Monte-Carlo sample from the true generating model $p_t$. We hence compute $n$ leave-one-out (LOO) posterior distributions $p(\theta|y_{-i})$, where $y_{-i}$ denotes the data with observation $y_i$ removed. Using the LOO posteriors, we can estimate the elpd in Eq.~\eqref{elpd} as 
\begin{align}
\label{elpd_loo}  
\overline{\text{elpd}}_\text{loo} = & \frac{1}{n} \sum^n_{i=1} \log p_M(y_i|y_{-i}) \\
= &\frac{1}{n} \sum^n_{i=1} \log \int p_M(y_i|\theta) p_M(\theta | y_{-i}) d\theta   \nonumber \\
= & \frac{1}{n} \, \text{elpd}_\text{loo} \nonumber\,, 
\end{align}
where $p_M(y_i | \theta)$ is the likelihood, and $p_M(\theta | y_{-i})$ is the posterior for $\theta$ where we hold out observation $y_i$. 

Although the many good properties of LOO, scaling the approach to large data is a problem. The naive approach to LOO means that $n$ posteriors need to be computed. In situations with large $n$, the cost of just computing one posterior may be large, hence leading to poor scaling.

\subsection{Approximating LOO}
A number of approximate techniques have recently been proposed to approximate exact LOO. \citet{wang2018approximate} and \citet{giordano2019swiss} propose LOO-approximations with very appealing error bounds for M-estimators. The main idea is to fit a model on the complete data set and then extrapolate to capture the effect of holding out individual observations using a second-order Taylor approximation. For some special classes of models, such as Gaussian processes, specialized LOO approximations have been proposed \citep{held2010posterior,vehtari2016bayesian}.

In the Bayesian domain, similar ideas was introduced by \citet{gelfand1996model} using self-normalized importance sampling (IS). The idea is to use the full posterior distribution as the proposal distribution in an importance sampling scheme with the LOO posterior as the target distribution. In this way, we only need to estimate the model once. Given $S$ draws from the full posterior $p(\theta|y)$, we can estimate the individual elpd contributions as 
\begin{align}
\label{eq:is_estimate}
\log \hat{p}(y_i|y_{-i}) = \log\left( \frac{\frac{1}{S} \sum_{s=1}^S p_M(y_i|\theta_s) r(\theta_s)}{\frac{1}{S} \sum_{s=1}^S r(\theta_s)} \right)\,,\\ \quad\quad r(\theta_s) = 
\frac{p_M(\theta_s|y_{-i})}{p_M(\theta_s|y)} 
\propto \frac{1}{p_M(y_i|\theta_s)} \,,
\end{align}
and where the last step is the result for factorizable likelihoods. In case of highly influential observations, the proposal distribution, i.e., the full posterior distribution can be very different than the target LOO posterior, and importance sampling estimates may have large variance. \citet{vehtari2015pareto} present Pareto-smoothed importance sampling (PSIS) to smooth the importance ratios $r(\theta_s)$, introducing a small bias, but reducing the overall mean-squared error. The PSIS approach also has the benefit that the estimated shape parameter $k$ of the generalized Pareto distribution can diagnose when the importance sampling approach has too large (or infinite) variance \citep{vehtari2017practical}. 

LOO is closely related to the Watanabe-Akaike or widely applicable information criterion \citep[WAIC, ][]{watanabe2010asymptotic}. The elpd of a given model can be estimated using WAIC as 
\begin{align}
\label{waic}
\text{elpd}_\text{WAIC} = &\sum_{i=1}^n \log p(y_i|y) - V_\theta(\log p(y_i|\theta)) \nonumber \\ = & \sum_{i=1}^n \log p(y_i|y) - p_{i,\text{eff}}\,,
\end{align}
where $V_\theta(\log p(y_i|\theta))$ is the variance of the log likelihood over the (full) posterior $p(\theta|y)$, often called the \emph{effective number of parameters} or $p_\text{eff}$. It has been shown that WAIC and LOO are asymptotically equivalent \citep{watanabe2010asymptotic}, but LOO has been found to be more robust than WAIC in the finite data domain, especially in the case of outliers or weak priors. This is because the WAIC approximation ignores higher order terms and these may be non-negligible for finite data \citep{gelman2014understanding,vehtari2016bayesian,vehtari2017practical}. Importantly, both LOO and WAIC are consistent estimators of the true $\text{elpd}$ under mild assumptions \citep{watanabe2010asymptotic}.

\subsection{LOO for Large Data}

In \citet{magnusson2019bayesian}, two problems with Bayesian PSIS-LOO for large data are addressed. First the results of \citet{gelfand1996model} in Eq.~\eqref{eq:is_estimate} are extended to approximate inference methods such as variational Bayes (VB) and Laplace approximations, and second an efficient subsampling method using the Hansen-Hurwitz \citep[HH, ][]{hansen1943} estimator is proposed. \citet{magnusson2019bayesian} use the full log predictive density $\log p(y_i|y)$ (lpd) and the $\log p(y_i|\hat{\theta})$, the point log predictive density (plpd) as an auxiliary variable, $\tilde{\pi}$. The data is then subsampled proportionally $\tilde{\pi}$ to efficiently estimate the elpd as
\begin{align}
\label{eq:hh_estimate}
\widehat{\text{elpd}}_\text{HH} = \frac{1}{m} \sum_{j\in\mathcal{S}}  \frac{1}{\tilde{\pi}_{j}} \log \hat{p}(y_j|y_{-j})\,,
\end{align}
where $m$ is the subsample size and $\mathcal{S}$ is the subsample.
This approach works well for estimating the elpd of individual models and has good theoretical properties, but it has two problems when used for model comparison.

First, when comparing models we are often interested in the elpd for a set of different models. Since the auxiliary information is used in the subsampling step, this means that we would need to draw a new subsample for each estimate of interest, such as (1) the elpd of each model, (2) the elpd difference between models, and (3) the variance of each elpd estimate. Ideally, we would like to just draw one subsample and then based on that subsample compute all estimates of interest. 

Second, using the $\log p(y_i|\hat{\theta})$ as the auxiliary variable misses the effect of the efficient number of parameters in the model $p_\text{eff}$, i.e., the model complexity, as can be seen in Eq.~\eqref{waic}. This means that we would need larger subsample sizes when estimating more complex models. 

\subsection{Contributions and Limitations}

In this paper, we focus on methods for scaling Bayesian LOO methods for comparing models for large data. We show that using the difference estimator combined with simple random sampling without replacement is very well suited for model comparison purposes. Since model auxiliary information is not used in the sampling stage, but in the estimation, the approach is much better suited for the situation of model comparison. 

We also show that incorporating estimates of $p_\text{eff}$ improves the performance of the subsampling and propose fast methods to approximate $p_\text{eff}$ for large data and propose computationally efficient approximations, $\tilde{\pi}$, that take $p_\text{eff}$ into account.

We prove that the difference estimator will converge in mean to the true LOO ($\overline{\text{elpd}}_\text{loo}$) for any LOO approximation $\tilde{\pi}$ that converge in mean to $\pi$, irrespective of subsample size and the number of draws from the posterior. We also prove that our proposed approximations will converge in mean to $\pi$.

Together this makes the approach well suited for generic large-data model inference, such as in probabilistic programming frameworks as Stan \citep{carpenter2017stan}. 

The limitations with the proposed approach are the same as using general PSIS-LOO \citep[see ][ for a detailed discussion]{vehtari2017practical}, such as that the likelihood needs to be factorizable for Eq.~\eqref{eq:is_estimate} to hold.

\section{LARGE DATA MODEL COMPARISON USING LOO}

Let $\text{elpd}_\text{A}$ and $\text{elpd}_\text{B}$ be the $\text{elpd}_\text{loo}$ for model A and model B, respectively. To compare models, we are interested in the difference in elpd between models, $ \text{elpd}_\text{D} = \text{elpd}_\text{A} - \text{elpd}_\text{B}$ as well as $V(\text{elpd}_\text{D})$, the variability due to the data, where
\[
V(\text{elpd}_\text{D})=V(\text{elpd}_\text{A}) + V(\text{elpd}_\text{B}) - 2\, \text{Cov}(\text{elpd}_\text{A}, \text{elpd}_\text{B}) \,.
\]

To efficiently estimate $V(\text{elpd}_\text{D})$, we propose to use the \emph{difference estimator} and simple random sampling without replacement (SRS). We also propose to include $p_\text{eff}$ in Eq.~\eqref{waic} for better approximations of $\log p(y_i|y_{-i})$. This makes it possible to better compare models by computing the full posterior distributions \emph{once} and then compare models performance on \emph{one} subsample of observations.

\subsection{The Difference Estimator}

Let $\pi_i = \log p(y_i|y_{-i})$ be our variable of interest where $\text{elpd}_\text{loo}=\sum^n_i \pi_i$. Then let $\tilde{\pi}_i$ be any approximation of $\log p(y_i|y_{-i})$. Given $\tilde{\pi}_i$ we can use the difference estimator, a special case of the regression estimator \citep[Ch. 7]{cochran77}, together with SRS. The $\text{elpd}_\text{loo}$ can then be estimated as
\begin{align}
\label{diff_est}
\widehat{\text{elpd}}_\text{diff,loo} = \sum^n_{i=1} \tilde{\pi}_i + \frac{n}{m} \sum_{j\in\mathcal{S}} \left(\pi_j - \tilde{\pi}_j\right)\,,
\end{align}
where $m$ is the subsampling size and $\mathcal{S}$ is the subsample. The (subsample) variance associated with the difference estimator is 
\begin{align}
\label{diff_est_var}
V(\widehat{\text{elpd}}_\text{diff,loo}) = n^2 \left(1-\frac{m}{n}\right) \frac{s_e^2}{m} \,,
\end{align}
where $s_e^2$ is the sample standard deviations of the approximation error $e_j = \pi_j - \tilde{\pi}_j$, i.e $s^2_e=\frac{1}{m-1}\sum_j^m (e_j - \bar{e})^2$ and $\bar{e}=\frac{1}{m} \sum_j^m e_j$. 

The proposed approach has two important properties. First, as the sequence of numbers $\tilde{\pi}_i \rightarrow \pi_i$, $V(\widehat{\text{elpd}}_\text{diff,loo}) \rightarrow 0$. Unlike the HH  estimator in Eq.~\eqref{eq:hh_estimate}, we also have the property that as $\frac{m}{n}\rightarrow 1$, $V(\widehat{\text{elpd}}_\text{diff,loo}) \rightarrow 0$. This finite correction factor increases the efficiency also in smaller data, where LOO still can be costly. 

Second, the main benefits of using the difference estimator is that we can use a sampling scheme that do not depend on the models. Instead, we use the \emph{same} subsample to estimate all properties of interest, such as $\text{elpd}_\text{loo}$ for all models. This reduce the computational cost, especially for model comparisons, since we can reuse the already computed values for the sample when computing the $\text{elpd}_D$. Similarly, for model comparison, we are also interested in estimating $V(\text{elpd}_\text{loo})=\sigma^2_\text{loo}$, the variability of the $\text{elpd}_\text{loo}$ and $\text{elpd}_\text{D}$, for comparing models. Using the difference estimator we estimate $\sigma^2_\text{loo}$ as 
\begin{align}
\label{sigma_loo_estimator}
\hat{\sigma}^2_\text{diff,loo} = & \sum_{i=1}^n \tilde{\pi}_i^2 + \frac{n}{m} \sum_{j\in\mathcal{S}} \left(\pi^2_j - \tilde{\pi}^2_j\right) +  \\ 
& \frac{1}{n} \left[\left(\frac{n}{m} \sum_{j\in\mathcal{S}} \left(\pi_j - \tilde{\pi}_j\right)\right)^2 - V(\widehat{\text{elpd}}_\text{diff,loo})\right] + \nonumber \\ & \frac{1}{n} \left[ 2 \left(\sum_{i=1}^n \tilde{\pi}_i\right) \widehat{\text{elpd}}_\text{diff,loo} - \left(\sum_{i=1}^n \tilde{\pi}_i\right)^2 \right] \,. \nonumber
\end{align}
Eq.~\eqref{sigma_loo_estimator} shows that using the difference estimator, we only need to compute $\tilde{\pi}^2_i$ to estimate $\sigma^2_\text{loo}$, using the same subsample. The difference estimator is hence better suited for the case of large data Bayesian model comparison. We conclude by noting that the difference estimator is unbiased.


\begin{proposition} \label{pa0}
The estimators $\widehat{\text{elpd}}_\text{diff,loo}$ and $\hat{\sigma}^2_\text{diff,loo}$ are unbiased with regard to $\text{elpd}_\text{loo}$ and $\sigma^2_\text{loo}$.
\end{proposition}

\begin{proof}
See the supplementary material.
\end{proof}

\emph{Remark} Note that $\hat{\sigma}^2_\text{diff,loo}$ is most often an optimistic estimate for the variability of $\text{elpd}_\text{loo}$, since no general unbiased estimator of the true variability exists \citep{bengio2004no}.

\subsection{Fast Approximate LOO Surrogates}
\label{sec:p_eff}

For the difference estimator to have small variance, we need good approximations of the variable of interest. We start with the following definition.

\begin{definition}
\label{consistent_approximation}
An approximation $\tilde{\pi}_i$ of $\pi_i$ is said to converge in mean if $\E|\pi_i - \tilde{\pi}_i| \ra 0$ as $n \ra \iy$.
\end{definition}

When estimating $\text{elpd}_\text{loo}$ we want the approximation $\tilde{\pi}_i$ to have the following three properties: 
\begin{enumerate}
    \item a good finite data approximation of $\pi_i$,
    \item computationally cheap, and
    \item converges in mean to $\pi_i$.
\end{enumerate}

The last property is needed for Proposition \ref{pa2} and \ref{pa3}, that shows favorable theoretical scaling characteristics of the estimator as $n \rightarrow \infty$.

The WAIC estimator in Eq.~\eqref{waic} indicates that using the plpd as $\tilde{\pi}_i$, such as in \citet{magnusson2019bayesian}, will essentially miss the effect of the effective number of parameters $p_\text{eff}=V_\theta(\log p(y_i|\theta))$ in approximating $\pi_i$. Since it has been shown by \citet{watanabe2010asymptotic} that the WAIC and LOO are asymptotically equivalent, including $p_\text{eff}$ will improve over the plpd, especially for more complex models. Using the WAIC as approximation we set $\tilde{\pi}_i = \log p(y_i|y) - V_\theta(\log p(y_i|\theta))$, where $\text{elpd}_\text{WAIC} = \sum_i^n \tilde{\pi}_i$ in Eq.~\eqref{waic}.

A problem with this approximation is that for each observation we need to integrate over the posterior to compute $\tilde{\pi}_i$ based on the WAIC in Eq.~\eqref{waic} and hence the approximation is more costly than using the plpd. To reduce the cost of computing $\tilde{\pi}$, the simplest way is to reduce the number of draws to approximate $\tilde{\pi}$ to $S_{\tilde{\pi}}$ where $S_{\tilde{\pi}} < S$, but this also reduces the accuracy. 

Another approach to approximate $\pi_i$ more computationally efficient is to approximate $p_{i,\text{eff}}$ in Eq.~\eqref{waic} directly using a Taylor approximation:
\begin{align}
\label{delta}
 \tilde{p}_{i,\text{eff}} = & \nabla \log(p(y_i|\theta))^T \Sigma_\theta \nabla \log(p(y_i|\theta)) \nonumber \\ & + \frac{1}{2} \tr(H_{i,\theta} \Sigma_\theta H_{i,\theta} \Sigma_\theta)\,,
\end{align}
where $\nabla \log(p(y_i|\theta))$ and $H_{i,\theta}$ are the gradient and Hessian of $\log(p(y_i|\theta))$ with respect to $\theta$, respectively. This gives us an approximation of the $p_{i,\text{eff}}$ without the need to compute $V_\theta(\log p(y_i|\theta))$ over all $S$ draws. We can use the idea to produce three different approximations, $\Delta_2\text{WAIC}$ that uses Eq.~\eqref{delta}, $\Delta_1\text{WAIC}$ that only use the first order (gradient) term and $\Delta_1\text{WAIC}_\text{m}$ that only uses the first order term and the marginal variances, i.e., using $\text{diag}(\Sigma_\theta)$ instead of $\Sigma_\theta$ in Eq.~\eqref{delta}.


Another approach to approximate PSIS-LOO is to use truncated importance sampling \citep[TIS, ][]{ionides2008truncated}. TIS-LOO will increase the bias but is less computationally costly since we remove the cost of estimating Pareto-$k$ and smoothing using the Pareto distribution. As has been shown in \citet{vehtari2015pareto}, TIS-LOO can approximate LOO better than WAIC at the same computational cost. As with WAIC, we can also use TIS-LOO with fewer posterior draws to compute computationally less costly approximations of $\tilde{\pi}$. 

\subsection{Summary of Approach}

The difference estimator and fast LOO surrogates lead us to how we can compare models for large data. 
\begin{enumerate}
    \item Compute the posterior $p_A(\theta|y)$ and $p_B(\theta|y)$ for model A and B, respectively.
    \item Compute $\tilde{\pi}$ for model A and B using an approximation that fulfill the properties in Sec. \ref{sec:p_eff}.
    \item Compute the approximate differences as $\tilde{\pi}_{i, D}=\tilde{\pi}_{i,A} - \tilde{\pi}_{i,B}$ for all $n$.
    \item Draw a subsample of size $m$ and compute $\pi_{j,D}=\pi_{j,A} - \pi_{j,B}$ for all $m$. 
    \item Estimate $\text{elpd}_\text{D}$ and $V(\text{elpd}_\text{D})$ using Eq.~\eqref{diff_est} and (\ref{diff_est_var}).
\end{enumerate}
Depending on the accuracy of the chosen approximation $\tilde{\pi}$, we can easily increase the subsampling size $m$ to reach the desired accuracy. 

\subsection{Asymptotic Properties}

Here we study the asymptotic properties of using the difference estimator together with any approximation $\tilde{\pi}_i$ that converges in mean to $\pi$. Let $(y_1,y_2,\ldots,y_n)$, $y_i \in \mathcal{Y} \subseteq \RR$ be drawn from a  true density $p_t = p(\cdot|\theta_0)$ with the true parameter $\theta_0$ that is assumed to be drawn from $p(\theta)$ on the parameter space $\Theta$, an open and bounded subset of $\RR^d$. To prove Proposition \ref{pa2} and \ref{pa3} we make the following assumptions:
\begin{itemize}
  \item[(i)] the likelihood $p(y|\theta)$ satisfies that there is a function $C:\mathcal{Y} \ra \RR_+$, such that $\E_{y \sim p_t}[C(y)^2] < \iy$ and such that for all $\theta_1$ and $\theta_2$, $|p(y|\theta_1)-p(y|\theta_2)| \leq C(y)p(y|\theta_2)\n\theta_1-\theta_2\n$.
   \item[(ii)] $p(y|\theta)>0$ for all $(y,\theta) \in \mathcal{Y} \times \Theta$,
   \item[(iii)] There is a constant $M<\iy$ such that $p(y|\theta) < M$ for all $(y,\theta)$,
  \item[(iv)] for all $\theta$, $\int_{\mathcal{Y}}(-\log p(y|\theta))p(y|\theta)dy < \iy$.  
  \item[(v)] all assumptions needed in the Bernstein-von Mises Theorem \citep{walker1969asymptotic}, and
  \item[(vi)] $p(\theta|y)>0$ or all $\theta \in \Theta$
\end{itemize}

Here we also generalize the definition of $r(\theta_s)$ in Eq.~\eqref{eq:is_estimate} to handle arbitrary posterior approximations \citep[see][ for an extended discussion]{magnusson2019bayesian}. Hence, let 
\begin{align*}
    r(\theta_s) \propto & \frac{1}{p(y_i|\theta_s)} \frac{p(\theta_s|y)}{q(\theta_s|y)}\,.
\end{align*}

Now, let $\widehat{\overline{\text{elpd}}}_{\rm diff,loo}=\frac{1}{n}\widehat{\text{elpd}}_{\rm diff,loo}$, we then have the following propositions.

\begin{proposition} \label{pa2}
  For any approximation $\tilde{\pi}_i$ that converges in mean to $\pi_i$, we have that $\widehat{\overline{\text{elpd}}}_{\rm diff,loo}$ converges in mean to $\overline{\text{elpd}}_{\rm loo}$.
\end{proposition}
\vspace{-0.5\baselineskip}
\begin{proof}
See the supplementary material.
\end{proof}

\begin{proposition} \label{pa3}
Let the subsampling size $m$ and the number of posterior draws $S$ be fixed at arbitrary integer numbers, let the data size $n$ grow, assume that (i)-(vi) hold and let $q=q_n(\cdot|y)$ be any consistent approximate posterior. Write $\hat{\theta}_q = \arg\max\{q(\theta): \theta \in \Theta\}$ and assume further that $\hat{\theta}_q$ is a consistent estimator of $\theta_0$. Then
\[\tilde{\pi}_i \ra \pi_i\]
in mean for any of the following choices of $\tilde{\pi}_i$, $i=1,\ldots,n$.
\begin{itemize}
  \item[(a)] $\tilde{\pi}_i = \log p(y_i|\hat{\theta}_q)$.
  \item[(b)] $\tilde{\pi}_i = \log p(y_i|y)+V_{\theta \sim p(\cdot|y)} (\log p(y_i|\theta))$.
  \item[(c)] $\tilde{\pi}_i = \log p(y_i|y) - \nabla \log p(y_i|\hat{\theta})^T \Sigma_\theta \nabla \log p(y_i|\hat{\theta})$ for any given fixed $\hat{\theta}$ and where the covariance matrix is with respect to $\theta \sim p(\cdot|y)$.
  \item[(d)] $\tilde{\pi}_i = \log p(y_i|y) - \nabla \log p(y_i|\hat{\theta})^T \Sigma_\theta \nabla \log p(y_i|\hat{\theta}) - \frac12 \rm{tr}(H_{\hat{\theta}} \Sigma_\theta H_{\hat{\theta}}) \Sigma_\theta)$ for any given fixed $\hat{\theta}$ and where the covariance matrix is as in (c)
  \item[(e)] $\tilde{\pi}_i = \log \hat{p}(y_i|y_{-i})$ as defined in (\ref{eq:is_estimate}).
\end{itemize}
\end{proposition}
\vspace{-0.5\baselineskip}
\begin{proof}
See the supplementary material.
\end{proof}
\vspace{-0.5\baselineskip}

Proposition \ref{pa2} and \ref{pa3} generalizes the scaling properties of \citet{magnusson2019bayesian}, namely that in the limit, we essentially only need a subsampling size of $m=1$ and $S=1$ draw from the posterior to estimate the $\overline{\text{elpd}}_\text{loo}$ exact using the difference estimator. Proposition \ref{pa2} show that for any $\tilde{\pi}$ that converges in mean to $\pi$ the difference estimator will converge in mean to the true $\overline{\text{elpd}}_\text{loo}$. Using Proposition \ref{pa3} we also have that the favorable scaling properties holds also for WAIC, our proposed approximations of WAIC, and using importance sampling for any choice of $S$. The results also hold for posterior approximations, as long as consistent posterior approximations are used, such as variational inference and Laplace approximations for regular models. In the supplementary material we also extend Proposition \ref{pa3} to additional choices of $\tilde{\pi}$.


\subsection{Computational Cost}

The computational cost of the proposed method will depend on the total number of observations for which we need to compute $\tilde{\pi}_i$ and hence, in most situations, computing $\tilde{\pi}$ will dominate. This makes it relevant to understand how these costs relate to the total number of parameters in the likelihood function $P$ (not the total number of parameters in the model) and the total number of posterior draws $S$. 
The overall cost for the different approximations proposed is presented in Table \ref{tab:approx_tilde}. Computing the full PSIS-LOO has the cost of $O(nPS)$, given that the evaluation of the log-likelihood is linear in $P$, i.e., the same complexity as WAIC, but with larger constants. Different trade-offs can be made depending on the specific likelihood where the approximation cost range from the cheapest, the plpd, to the most costly, WAIC/TIS with a large number of posterior draws $S$. The plpd only computes the log-likelihood once, while the full WAIC/TIS approach needs to compute it $S$ times.

\begin{table}[t]
  \centering
  \begin{tabular}{lll}
    $\tilde{\pi}$     & Needs     & Cost \\
    \midrule
    plpd & $\hat{\theta}$ & $O(nP)$     \\
    $\text{TIS}_\text{S}$ & $S$ draws from $p(\theta|y)$  & $O(nPS)$     \\    
    $\text{WAIC}_\text{S}$ & $S$ draws from $p(\theta|y)$  & $O(nPS)$     \\    $\Delta_1\text{WAIC}_\text{m}$ & $\nabla \log p(y_i|\theta)$, $\hat{\theta}$ and $\text{tr}(\Sigma_{\theta})$ & $O(nP)$  \\    
    $\Delta_1\text{WAIC}$ & $\nabla \log p(y_i|\theta)$, $\hat{\theta}$ and $\Sigma_{\theta}$ & $O(nP^2)$     \\
    $\Delta_2\text{WAIC}$  & $\nabla \log p(y_i|\theta)$, $H_\theta$, $\hat{\theta}$ and $\Sigma_{\theta}$ & $O(nP^3)$     \\
  \end{tabular}
  \caption{Computational costs for approximations of $\pi$.}
    \label{tab:approx_tilde}
\end{table}

\begin{table}[t]
\centering
\begin{tabular}{ll}
$M$ & Description \\ 
  \midrule
 1 & Full pooling \\ 
 2 & Partial pooling \\ 
 3 & No pooling \\ 
 4 & Variable intercept \\
 5 & Variable slope \\ 
 6 & Variable intercept and slope \\  
\end{tabular}
\caption{Radon models. For full model specification, see supplementary material.} 
\label{tab:radon_models}

\end{table}

\section{EXPERIMENTS}

We study the proposed method using both simulated and real data.  We simulated datasets with $10^4$, $10^5$, and $10^6$ observations to fit Bayesian linear regression (BLR) models. The simulated data is generated with different signal-to-noise ratios resulting in $R^2$ values of approximately 0.1, 0.5 and 0.9. To simulate sparse regression, we generated another dataset with only one covariate with $\beta \neq 0$. See supplementary material for the details of simulations. As the first real data, we use the radon data of \citet{lin1999analysis}. The dataset consists of roughly 12~000 radon level measurements in 400 counties (groups). Table \ref{tab:radon_models} lists different non-hierarchical and hierarchical models used. The exact model specification with Stan code can be found in the supplementary material. To compare different logistic models with simple linear effects, interaction effects and splines, we use the arsenic wells data of \citet{gelman2006data} with 3~020 observations. The datasets are big enough to demonstrate the most important properties, but small enough to be easily fit using MCMC as a gold standard.

We use Stan \citep{carpenter2017stan,standev2018stancore} for inference using 4 chains, a sample size of 2~000, a warmup of 1~000 iterations and a dynamic HMC algorithm \citep{hoffman2014no,Betancourt2017} and the \texttt{rstanarm} R package \citep{rstanarm} to fit spline models. Convergence diagnostics were made using $\hat{R}$ diagnostic \citep{Gelman+etal+BDA3:2013} and HMC specific diagnostics \citep{Betancourt2017}. For simplicity, in this paper we limit the scope to MCMC, but the proposed approach is trivial to use with any consistent posterior approximation using the approximation correction of \citet{magnusson2019bayesian}. Our approach has been implemented based on the \texttt{loo} R package \citep{vehtari2018loo} framework for Stan and is available at \url{https://cran.r-project.org/package=loo}.

\begin{table*}[t]
\centering
\begin{tabular}{llrrrrrrrr}
Data & $M$ & SE($\widehat{\text{elpd}}_\text{diff,WAIC2k}$) & SE($\widehat{\text{elpd}}_\text{HH,plpd}$)  & SE($\widehat{\text{elpd}}_\text{HH,WAIC2k}$) & SE($\widehat{\text{elpd}}_\text{SRS}$) & $\hat{\sigma}_\text{loo}$ \\ 
  \midrule
Radon & 1 & 0.002 & 0.7  & 0.001 & 949 & 88 \\ 
   & 2 & 1.0 & 42  & 0.81 & 1012 & 94 \\ 
   & 3 & 9.2 & 74  & 6.4 & 1012 & 94 \\    
   & 4 & 1.0 & 40 & 0.79 & 1013 & 94 \\ 
  & 5 & 13 & 84  & 12 & 972 & 90 \\   
   & 6 & 10 & 97 & 15 & 1050 & 96 \\ 
  \midrule   
  $R^2=0.9$& BLR  & 0.03 & 4.0 & 0.02 & 704 & 70 \\ 
  $R^2=0.5$&  & 0.04  & 4.9  & 0.03 & 711 & 72 \\ 
  $R^2=0.1$&  & 0.04 & 5.8  & 0.03  & 756 & 76 \\ 
\end{tabular}
\vspace{-0.2\baselineskip}
\caption{The subsampling standard error (SE) for individual models of different subsampling approaches to estimate $\text{elpd}_\text{loo}$ with $m=100$. The results are averaged over 100 set of subsamples. $\hat{\sigma}_{\text{loo}}$ is the estimated standard deviation of the LOO-CV for comparison.} 
\label{tab:hh_vs_diff_performance}
\vspace{-0.4\baselineskip}
\end{table*}

With the empirical evaluations, we study the following research questions: (1) does using better approximations of $\pi$ improve the empirical performance, (2) which approximation $\tilde{\pi}$ should be preferred, (3) how does the difference estimator compare with the HH approach, and (4) how well does the method scale for large-data model comparison?

\subsection{Performance}

Table \ref{tab:hh_vs_diff_performance} shows the estimator variance when including $p_\text{eff}$ in the estimation, by comparing our proposed approach with the HH method proposed in \citet{magnusson2019bayesian}. From the results we can see the benefit of including $p_\text{eff}$ in the estimation of $\text{elpd}_\text{loo}$. Including $p_\text{eff}$ improves the estimation by orders of magnitude compared to using only the plpd, both for the HH as well as for the difference estimator, both in turn improve with orders of magnitude over simple random sampling. Table \ref{tab:approx} shows similar results where we see that all approximations of $\tilde{\pi}$ that include $p_\text{eff}$ improves over using just the plpd as approximate surrogate. 

The downside of using the $\text{WAIC}_{2k}$, WAIC based on 2~000 draws, as the approximate surrogate variable is that it is computationally costly (see Table \ref{tab:approx_tilde}). Hence, it is of interest to study the performance of the different surrogate approximations shown in Table \ref{tab:approx}. We can see that any approximation is better than just using the plpd, as noted above. Using $\Delta_1 \text{WAIC}_\text{m}$, that have the same computational complexity as the plpd, is better in all models, showing the benefit of including $p_{\text{eff}}$ in $\tilde{\pi}_i$. Table \ref{tab:approx_tilde} shows that as we improve the approximation we get better and better estimates of the $\text{elpd}_\text{loo}$, even though the benefit of the better approximations varies from model to model. TIS, computed by averaging 2~000 posterior draws, is the most accurate approach. Hence, we are confronted with the trade-off between sampling size and cost of computing $\tilde{\pi}$. 


\begin{table*}[t]
\centering
\begin{tabular}{rrrrrrrrrrr}
$M$ & $\text{WAIC}_{2k}$ & $\text{TIS}_{2k}$ & $\text{WAIC}_{100}$ & $\text{TIS}_{100}$ & $\Delta_2 \text{WAIC}$ & $\Delta_1 \text{WAIC}$ & $\Delta_1 \text{WAIC}_\text{m}$ & plpd\\ 
  \midrule
   1 & 0.0 & 0.0 & 1.6 & 1.6 & 0.5 & 0.5 & 0.6 & 1 \\ 
   2 & 1.0 & 0.2 & 21 & 20 & 30 & 31 & 31 & 53 \\ 
   3 & 9.2 & 1.7 & 29 & 29 & 45 & 56 & 56 & 87 \\ 
   4 & 1.0 & 0.3 & 22 & 22 & 29  & 30 & 30 & 51 \\ 
   5 & 13 & 9.8 & 26 & 36 & 32 & 39 & 40 & 81 \\
   6 & 10 & 7.5 & 34 & 42 & 50 & 57 & 90 & 107 \\

\end{tabular}
\vspace{-0.4\baselineskip}
\caption{SE($\text{elpd}_\text{diff}$) using different approximations $\tilde{\pi}$ for the Radon data with $m=100$. The results are averaged over 100 set of subsamples.} 
\label{tab:approx}
\vspace{-1\baselineskip}
\end{table*}

Table \ref{tab:hh_vs_diff_performance} contains the results of both the HH estimator and the difference estimator with $\text{WAIC}_{2k}$ as $\tilde{\pi}$. The performance between the estimators are similar, but in most cases, the HH estimator is marginally better. This can be explained by the HH estimator being better at capturing heteroscedastic approximation errors
\citep[see][Ch. 7,9A]{cochran77}. Since we expect the approximation to be worse for smaller values of $\pi_i$, this explains the difference in performance. 


\begin{table}[ht]
\centering
\begin{tabular}{rrrrrrr}
$n$ & $M$ & $\tilde{\pi}$ & $\widehat{\text{elpd}}_\text{D}$ & SE($\widehat{\text{elpd}}_\text{D}$) \\ 
  \midrule
$10^5$ & RHS vs. N & $\text{plpd}$ & -47.6 & 5.2  \\ 
 & RHS vs. N & $\text{TIS}_{10}$ & -52.3 & 0.04  \\ 
  \midrule
$10^6$ & RHS vs. N & $\text{plpd}$ & -44.7 & 7.0  \\ 
 & RHS vs. N & $\text{TIS}_{10}$ & -49.3 & 0.04  \\ 
\end{tabular}
\caption{Comparing models with Normal (N) and regularized horse-shoe (RHS) prior using a subsampling size of $m=100$ and the difference estimator. The $\widehat{\text{elpd}}_D$ is the estimated difference in $\text{elpd}_\text{loo}$ between models with the subsampling SE of the estimate. $\hat{\sigma}_D \approx 9$ for all.}  
\label{tab:model_compare_hs}
\vspace{-0.4\baselineskip}
\end{table}


\begin{figure}[t]
    \centering
    \includegraphics[width=0.8\columnwidth]{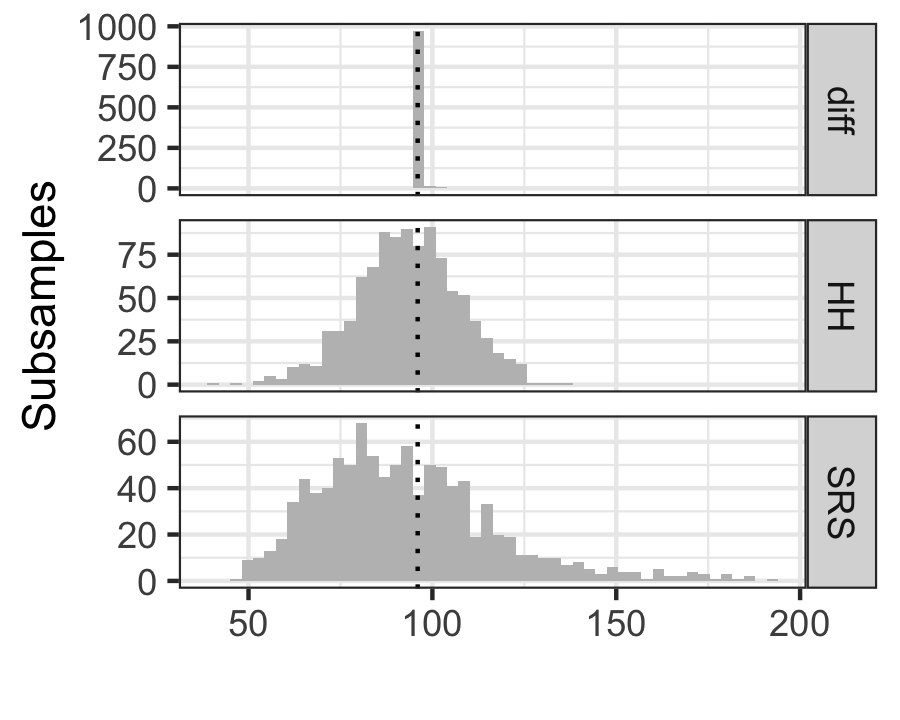}
    \vspace{-\baselineskip}
    \caption{1~000 estimates of $\sigma_\text{loo}$ for Radon Model 6 using $\text{TIS}_{2k}$ as $\tilde{\pi}$. True value is 96 (dotted line).}
    \label{fig:sigma_est}
    \vspace{-1.8\baselineskip}
\end{figure}

Figure \ref{fig:sigma_est} shows the benefit of the difference estimator in the estimation of $\sigma_\text{loo}$. It is quite clear that the difference estimator is much more efficient in estimating $\sigma_\text{loo}$, and that this efficiency comes from using both $\tilde{\pi}$ and $\tilde{\pi}^2$ as auxiliary variables in estimating $\sigma_\text{loo}$. To get the same efficiency for the HH estimator we would need to draw an additional subsample proportional to $\tilde{\pi}^2$ to reach a similar performance. Taken together, the HH estimator is marginally better in estimating elpd values for individual models. Although, the benefit of using $\tilde{\pi}_i$ in the estimation, instead of in the subsampling, is important in estimating $\sigma^2_\text{loo}$ as well as when comparing models.


\begin{table}[t]
\centering
\begin{tabular}{lrrrrrrrrr}
 $m$ & $M$ & $\widehat{\text{elpd}}_\text{D}$ & SE($\widehat{\text{elpd}}_\text{D}$) & $\hat{\sigma}_{D}$ \\ 

  \midrule
   100 & 6 vs. 4 & -233 & 69 & 22  \\
   & 6 vs. 2 & -303 & 35 & 26  \\
   & 6 vs. 3 & -333 & 57 & 25  \\
   & 6 vs. 5 & -1451 & 32 & 51 \\    
   & 6 vs. 1 & -1778 & 35 & 57 \\
   \midrule
   400 & 6 vs. 4 &  -236 & 22 & 24 \\ 
   & 6 vs. 3 &  -294 & 20 & 26 \\  
   & 6 vs. 2 &  -298 & 16 & 27 \\
   & 6 vs. 5 &  -1466 & 13 & 52 \\   
   & 6 vs. 1 &  -1780 & 13 & 58 \\ 
\end{tabular}
\caption{Comparing models using a subsample of size $m=100,400$, the difference estimator, and $\text{TIS}_\text{100}$ as approximation. The $\widehat{\text{elpd}}_D$ is the estimated difference in $\text{elpd}_\text{loo}$ between models with the subsampling SE of the estimate and $\hat{\sigma}_D$. $\hat{\sigma}_D$ is the estimated standard deviation of the $\text{elpd}_D$ for comparison. The \emph{naive} $\hat{\sigma}_D$ estimate is approximately 130 for all models} 
\label{tab:model_compare_radon_100_400}
\vspace{-1.2\baselineskip}
\end{table}




\subsection{Model Comparison}

Table \ref{tab:model_compare_hs} shows a large-scale example of a Bayesian linear regression model with 100 covariates and 1 million and 100~000 simulated data points with only one $\beta \neq 0$ to compare between a normal prior and the regularized horseshoe shrinkage prior \citep{piironen2017sparsity}. Using just the plpd we get a good approximation of $\pi$ so a subsample of $m=100$ is again sufficient to estimate the $\text{elpd}$ with sufficient accuracy. Using a better LOO surrogate, such as TIS with only 10 posterior draws, increases the accuracy considerably for these large datasets, without much additional computational cost in computing $\tilde{\pi}$. Table \ref{tab:model_compare_hs} also shows the positive scaling characteristics, the size of subsample needed to compare models does not change much, even though the total number of folds in LOO is increased tenfold.

Table \ref{tab:model_compare_radon_100_400} shows an example of comparing the different models for the Radon data based on an individual subsample. Using only a subsample of size 100 we can roughly identify which models should be preferred. The estimated $\sigma_\text{D}$ for difference in elpd compared to the reference model is much smaller than a naive approach where only the $\sigma_\text{loo}$ of individual models are used, i.e.,  $\sigma^2_\text{naive} = V(\text{elpd}_A) + V(\text{elpd}_B)$. In this case, we use the best (and most complex) model as a reference with a subsample of size 100. Using a slightly bigger subsample size of 400, the subsampling uncertainty is reduced further and we can conclude that model 6 has the best predictive performance. Increasing the subsampling size is not very costly once $\tilde{\pi}$ has been computed.

With earlier approaches, these comparisons would be much more costly. Either we would need to compute $\pi_i$ for all models and observations (full LOO), or if we use the HH estimator, we would need to draw a new set of observations for each model, each model comparison, and if we want an efficient estimate of $\sigma_D$, two sets of observations per model comparison.


In Table \ref{tab:model_compare_arsenic} we compare a generalized additive spline logistic model (GAM) with linear models with and without interactions. For these models, using  $\mathrm{TIS}_{100}$ as approximation, we needed to increase the subsample size to 300 to compare the model performance. A small subsample size is sufficient as the model does not have a complex hierarchical structure and hence the approximation works well to compare even small differences in $\text{elpd}_D$.

\begin{table}[t]
\centering
\begin{tabular}{lrrrr}
Model & $\widehat{\text{elpd}}_\text{D}$ & SE($\widehat{\text{elpd}}_\text{D}$) & $\hat{\sigma}_{D}$ \\ 
  \midrule
  GAM vs. interaction & -21 & 1.4 & 7.2 \\ 
  GAM vs. linear & -29 & 1.2 & 7.8 \\ 
\end{tabular}
\caption{Comparing arsenic models using $m=300$ using the difference estimator and the $\text{TIS}_\text{100}$ as approximation. The $\widehat{\text{elpd}}_D$ is the estimated difference in $\text{elpd}_\text{loo}$ between models with the subsampling SE of the estimate. $\hat{\sigma}_D$ is the estimated standard deviation of the $\text{elpd}_D$.} 
\label{tab:model_compare_arsenic}
\vspace{-0.1\baselineskip}
\end{table}



\begin{table}[t]
\centering
\begin{tabular}{lrrrr}
Model & $\widehat{\text{elpd}}_\text{D}$ & SE($\widehat{\text{elpd}}_\text{D}$) & $\hat{\sigma}_{D}$ \\ 
  \midrule
  $D=100$ vs. $D=101$ & -0.5  & 0.03 & 0.6\\ 
  $D=100$ vs. $D=110$ & -4.3  & 0.04 & 3.3\\ 
  $D=100$ vs. $D=99$ & -9.7 & 0.04 & 5.1\\ 
  $D=100$ vs. $D=90$ & 47.7 & 0.02 & 11.8\\   
\end{tabular}
\caption{Comparing BLR models with different number of included covariates for $R^2\approx0.1$, using a subsample of size $m=100$ using the difference estimator and $\mathrm{TIS}_{100}$ as approximation. The $\widehat{\text{elpd}}_D$ is the estimated difference in $\text{elpd}_\text{loo}$ between models with the subsampling SE of the estimate. $\hat{\sigma}_D$ is the estimated standard deviation of the $\text{elpd}_D$.} 
\label{tab:model_compare_blr1000}
\vspace{-.6\baselineskip}
\end{table}


  

Finally, Table \ref{tab:model_compare_blr1000} contains another simulated example with 100 covariates ($\beta = 1$ and $R^2 \approx 0.1$) and study the effect of adding irrelevant and removing relevant covariates for $n=10~000$. Again we see that using only $m=100$ we can get accurate estimates for $\text{elpd}_D$. The results of Table \ref{tab:model_compare_blr1000} also show the well-known inconsistency of LOO \citep{shao1993linear} in selecting the most parsimonious model, so for feature selection other approaches should be used \citep[see][]{piironen2018projective}.

\section{Conclusions}

Comparing different models is an important, but often overlooked, part of the process of predictive modeling. We propose a method for comparing and choosing between probabilistic models that are well suited to Bayesian model comparison for large data. First, using the difference estimator is much better suited to the setting of large-data model comparison. By using $\tilde{\pi}$ in the estimation rather than in sampling we reduce considerably the subsample size needed compared with approaches such as \citet{magnusson2019bayesian}, both when comparing models and estimating $\sigma_\text{loo}$ and $\sigma_\text{D}$. Second, including the number of efficient parameters in the auxiliary variable when estimating the elpd improves with orders of magnitude over not using this information. But using the better surrogate approximations is costly and introduce an accuracy-computational cost trade-off. If the gradient of the likelihood with respect to likelihood parameters is available, $\Delta_1\text{WAIC}_\text{m}$ can be used to improve performance without an additional computational cost compared to plpd. In all, we recommend using plpd as approximation for simpler models, while $\text{TIS}_{100}$ is recommended when comparing more complex hierarchical models.

We should not be too greedy and choose a very small subsample. A subsample that is too small, such as $m=10$, may, due to randomness, miss observations for which the approximation is bad, but still are not too uncommon among the observations. 

There are many additional possible improvements that we leave as future work. First, here we use the difference estimator with a simple random sample. We can use more or less any sampling strategy, such as stratifying the data based on the design matrix or by identified difficult observations. This can be further be used by adaptive optimal allocations between these strata. Second, we can further study other approaches to approximate LOO efficiently, here \citet{wang2018approximate,giordano2019swiss} are promising methods. As long as the approximation will converge in mean to $\pi$ the theory holds, making the method general as well as highly tunable for the specific problem at hand. 

In all, we propose a scalable method for fast model comparisons in case of large data. 

\newpage

\bibliographystyle{plainnat} 
\bibliography{references}

\newpage

\onecolumn

\setcounter{proposition}{0}

\section*{SUPPLEMENTARY MATERIAL}

\section*{Proofs}
\label{proofs}

The main quantity of interest is the mean expected log pointwise predictive density, which we want to use for model evaluation and comparison.

\begin{definition}[$\overline{elpd}$]
\label{definition}
The \emph{mean expected log pointwise predictive density} for a model $p$ is defined as
\[
\overline{\mathrm{elpd}} = \int p_t(x) \log p(x) \, dx
\]

where $p_t(x) = p(x|\theta_0)$ is the \emph{true} density at a new unseen observation $x$ and $\log p(x)$ is the log predictive density for observation $x$.

\end{definition}

We estimate $\overline{\mathrm{elpd}}$ using {\em leave-one-out cross-validation (loo)}.

\begin{definition}[Leave-one-out cross-validation]
\label{loo}
The loo estimator $\overline{\mathrm{elpd}}_{\rm loo}$ is given by

\begin{equation}
\label{elpdloo}
\overline{\mathrm{elpd}}_\mathrm{loo} = \frac{1}{n} \sum^n_{i = 1} \pi_i,
\end{equation}

where $\pi_i=\log p(y_i|y_{-i}) = \int \log p(y_i|\theta)p(\theta|y_{-i})d\theta$.

\end{definition}

To estimate $\overline{elpd}_{\rm loo}$ in turn, we use difference estimator. Definitions follow.

\begin{definition}
\label{diffloo}
Let $\tilde{\pi}_i$ be any approximation of $\pi_i$. The difference estimator of $\overline{elpd}_{\rm loo}$ based on $\tilde{\pi}_i$ is given by

\[
\widehat{\overline{\mathrm{elpd}}}_{\rm loo,diff} = \frac{1}{n}\left(\sum_{i=1}^{n}\tilde{\pi}_i + \frac{n}{m}\sum_{j\in\mathcal{S}}(\pi_j-\tilde{\pi}_j)\right),
\]
where $\mathcal{S}$ is the subsample set, $m$ is the subsampling size, and the probability of subsampling observation $i$ is $1/n$, i.e.\ the subsample is uniform with replacement.

\end{definition}

One important estimator of $\pi_i$ among others is the importance sampling estimator
\begin{equation} \label{ephat}
\log \hat{p}(y_i|y_{-i}) = \log\left( \frac{\frac{1}{S} \sum_{s=1}^S p(y_i|\theta_s) r(\theta_s)}{\frac{1}{S} \sum_{s=1}^S r(\theta_s)} \right)\,,
\end{equation}

where $r(\theta)$ is any suitable weight function such that $0 < r(\theta) < \iy$ for all $\theta \in \Theta$ and $(\theta_1,\ldots,\theta_S)$ is a sample from a suitable approximation of the posterior $p(\theta|y)$. We are in particular interested in the weight function

\begin{align} \label{weights}
    r(\theta_s) = & \frac{p(\theta_s|y_{-i})}{p(\theta_s|y)} \frac{p(\theta_s|y)}{q(\theta_s|y)} \nonumber \\
    \propto & \frac{1}{p(y_i|\theta_s)} \frac{p(\theta_s|y)}{q(\theta_s|y)}
\end{align}

and where $q(\cdot|y)$ is an approximation of the posterior distribution that satisfies for each $y$ that $q(\theta|y)$ iff $\theta \in \Theta$ , $\theta_s$ is a sample point from $q$ and $S$ is the total posterior sample size. (The condition on $q$ makes sure that $0 < r(\theta) < \iy$ for all $\theta$.)

In the case of truncated importance sampling, we instead truncate these weights and replace $r$ with $r_\tau$ given by 
\begin{align} \label{trunc_weights}
    r_\tau (\theta_s) = \min(r (\theta_s), \tau)\,,
\end{align}
where $\tau>0$ is the weight truncation \citep[see][for a more elaborate discussion on the choice of $\tau$]{ionides2008truncated}.  

\subsection*{Proof of Proposition 1}

\begin{proposition} \label{pa0}
The estimators $\widehat{\mathrm{elpd}}_\mathrm{diff}$ and $\hat{\sigma}^2_\mathrm{loo}$ are unbiased with regard to 
$\mathrm{elpd}_\mathrm{diff}$ and $\sigma^2_\mathrm{loo}$.
\end{proposition}

\begin{proof}

We start out by proving unbiasedness for the general estimator. Write the difference estimator as
\[
\widehat{\mathrm{elpd}}_\mathrm{loo,diff} = \sum_{i=1}^{n}\tilde{\pi}_i + \frac{n}{m}\sum_{i=1}^{n}\sum_{j\in\mathcal{S}}I_{ij}(\pi_j-\tilde{\pi}_j),
\]
where $I_{ij}$ is the indicator that data point $i$ is chosen as the $j$'th point of the subsample. Since $\E[I_{ij}]=1/n$, the expectation of the double sum is $\sum_i(\pi_i-\tilde{\pi}_i)$ and $\E[\widehat{\mathrm{elpd}}_\mathrm{loo,diff}] = \sum_i\pi_i$ as desired.

Next we prove unbiasedess of $\hat{\sigma}^2_\mathrm{loo,diff}$. We are interested in estimating the finite sampling variance using the difference estimator. This can be done as 

\begin{align}
\sigma^2_{\rm loo} = & \frac{1}{n} \sum^n_{i=1} (\pi_i -\bar{\pi})^2 \\
= & \frac{1}{n} \underbrace{\sum^n_{i=1} \pi_i^2}_a - {\underbrace{\left(\frac{1}{n}\sum^n_{i=1} \pi_i\right)}_b}^2
\end{align}

We can estimate $a$ and $b$ separately as follows. The first part can be estimated using the difference estimator with $\tilde{\pi}_i^2$ as auxiliary variable. Let $t_{\epsilon} = \sum^n_i \epsilon_i = \sum^n_i {\pi}^2_i - \tilde{\pi}^2_i =  t_{\pi^2} - t_{\tilde{\pi}^2}$, the we can estimate $a$ as

\[
\hat{a} = \frac{1}{n} (t_{\tilde{\pi}^2} + \hat{t}_{\epsilon})\,,
\]
where 
\[
\hat{t}_{\epsilon} = \frac{n}{m} \sum_{j\in\mathcal{S}} \left(\pi^2_j - \tilde{\pi}^2_j\right)\,.
\]

From the previous section, it follows directly that 

\[
E(\hat{a}) = \frac{1}{n} t_{\pi^2} = \frac{1}{n} \sum^n_{i=1} \pi_i^2,
\]

The second part, $b$, can then be estimated as 

\begin{align}
\hat{b} = & \frac{1}{n^2} \left[\hat{t}^2_e - v(\hat{t}_e) + 2 t_{\tilde{\pi}} \hat{t}_{\pi} - t_{\tilde{\pi}}^2 \right]\,,
\end{align}

with the expectation

\begin{align}
E(\hat{b}) = & \frac{1}{n^2} \left[E(\hat{t}^2_e) - E(v(\hat{t}_e)) + 2 t_{\tilde{\pi}} E(\hat{t}_{\pi}) - t_{\tilde{\pi}}^2 \right]\\
= & \frac{1}{n^2} \left[V(\hat{t}_e) + E(\hat{t}_e)^2 - V(\hat{t}_e) + 2 t_{\tilde{\pi}} t_{\pi} - t_{\tilde{\pi}}^2 \right]\\
= & \frac{1}{n^2} \left[t_e^2 + 2 t_{\tilde{\pi}} t_{\pi} - t_{\tilde{\pi}}^2 \right]\\
= & \frac{1}{n^2} \left[(t_{\pi}-t_{\tilde{\pi}})^2 + 2 t_{\tilde{\pi}} t_{\pi} - t_{\tilde{\pi}}^2 \right]\\
= & \frac{1}{n^2} t_{\pi}^2 = \left(\frac{1}{n} \sum_i^n {\pi_i}\right)^2
\end{align}

Using that

\begin{align}
E(v(\hat{t}_e)) = n^2 \left(1 - \frac{m}{n}\right) \frac{E(s^2_e)}{m} = n^2 \left(1 - \frac{m}{n}\right) \frac{S^2_e}{m} =V(\hat{t}_e)\,.
\end{align}

Combining the results we have that

\begin{align}
E(\hat{a} - \hat{b}) = \frac{1}{n} \sum^n_{i=1} \pi_i^2 - {\left(\frac{1}{n}\sum^n_{i=1} \pi_i\right)}^2 = \sigma_\mathrm{loo}^2 \,.
\end{align}

\end{proof}

{\bf Remark.} We believe this has probably been proven before, and hence this is probably not a new theoretical result.


\subsection*{Proof of Proposition 2 and 3}

The proof follows, in general, the proof of \citet{magnusson2019bayesian}. A generic Bayesian model is considered; a sample $(y_1,y_2,\ldots,y_n)$, $y_i \in \mathcal{Y} \subseteq \RR$, is drawn from a true density $p_t = p(\cdot|\theta_0)$ for some true parameter $\theta_0$. The parameter $\theta_0$ is assumed to be drawn from a prior $p(\theta)$ on the parameter space $\Theta$, which we assume to be an open and bounded subset of $\RR^d$.

\medskip

Several conditions are used. They are as follows.

\begin{itemize}
  \item[(i)] the likelihood $p(y|\theta)$ satisfies that there is a function $C:\mathcal{Y} \ra \RR_+$, such that $\E_{y \sim p_t}[C(y)^2] < \iy$ and such that for all $\theta_1$ and $\theta_2$, $|p(y|\theta_1)-p(y|\theta_2)| \leq C(y)p(y|\theta_2)\n\theta_1-\theta_2\n$.
   \item[(ii)] $p(y|\theta)>0$ for all $(y,\theta) \in \mathcal{Y} \times \Theta$,
   \item[(iii)] There is a constant $M<\iy$ such that $p(y|\theta) < M$ for all $(y,\theta)$,
  \item[(iv)] all assumptions needed in the Bernstein-von Mises (BvM) Theorem \citep{walker1969asymptotic},
  \item[(v)] for all $\theta$, $\int_{\mathcal{Y}}(-\log p(y|\theta))p(y|\theta)dy < \iy$.
\end{itemize}

{\bf Remarks.} 
\begin{itemize}
\item There are alternatives or relaxations to (i) that also work. One is to assume that there is an $\alpha>0$ and $C$ with $\E_y[C(y)^2]<\iy$  such that $|p(y|\theta_1)-p(y|\theta_2)| \leq C(y)p(y|\theta_2)\n\theta_1-\theta_2\n^\alpha$. 
There are many examples when (i) holds, e.g.\ when $y$ is normal, Laplace distributed or Cauchy distributed with $\theta$ as a one-dimensional location parameter.

\item The assumption that $\Theta$ is bounded will be used solely to draw the conclusion that $\E_{y,\theta}\n \theta-\theta_0\n \ra 0$ as $n \ra \iy$, where $y$ is the sample and $\theta$ is either distributed according to the true posterior (which is consistent by BvM) or according to a consistent approximate posterior. The conclusion is valid by the definition of consistency and the fact that the boundedness of $\Theta$ makes $\n\theta-\theta_0\n$ a bounded function of $\theta$. If it can be shown by other means for special cases that $\E_{y,\theta}\n\theta-\theta_0\n \ra 0$ despite $\Theta$ being unbounded, then our results also hold.
    
\end{itemize}

\begin{proposition}
  For any approximation $\tilde{\pi}_i$ that converges in $L^1$ to $\pi_i$, we have that $\widehat{\overline{\mathrm{elpd}}}_\mathrm{loo,diff}$ converges in $L^1$ to $\overline{\mathrm{elpd}}_\mathrm{loo}$.
\end{proposition}

\begin{proof}
For convenience we will write $\hat{e}:=\widehat{\overline{\mathrm{elpd}}}_\mathrm{loo,diff}$, which for our purposes is more usefully expressed as
\[\hat{e}=\frac{1}{n}\left(\sum_{i=1}^{n}\log \tilde{\pi}_i + \frac{n}{m}\sum_{i=1}^{n}\sum_{j=1}^{m}I_{ij}(\pi_i - \tilde{\pi}_i)\right),\]
where $I_{ij}$ is the indicator that sample point $y_i$ is chosen in draw $j$ for the subsample used in $\hat{e}$.

We then get, with respect to all randomness involved (i.e.\ the randomness in generating $y$ and the randomness in choosing the subsample in $\hat{e}$)
\begin{align*}
\E|\hat{e} - \overline{\mathrm{elpd}}_\mathrm{loo}| 
&\leq \frac{1}{n}\E\left[\sum_{1}^{n} |\tilde{\pi}_i - \pi_i| +  \frac{n}{m}\sum_{i=1}^{n}\sum_{j=1}^{m}I_{ij}|\pi_i - \tilde{\pi}_i|\right] \\
&= \E|\log \tilde{\pi}_i - \pi_i| + \frac{1}{m}\sum_{i=1}^{n}\sum_{j=1}^{m}\frac{1}{n}\E|\pi_i-\tilde{\pi}_i| \\
&= 2\E|\tilde{\pi}_i-\pi_i| \\
&\ra 0.
\end{align*}
\end{proof}

\begin{proposition} \label{pa}
Let the subsampling size $m$ and the number of posterior draws $S$ be fixed at arbitrary integer numbers, let the sample size $n$ grow, assume that (i)-(vi) hold and let $q=q_n(\cdot|y)$ be any consistent approximate posterior. Write $\hat{\theta}_q = \arg\max\{q(\theta): \theta \in \Theta\}$ and assume further that $\hat{\theta}_q$ is a consistent estimator of $\theta_0$. Then
\[\tilde{\pi}_i \ra \pi_i\]
in $L^1$ for any of the following choices of $\pi_i$, $i=1,\ldots,n$.

\begin{itemize}
  \item[(a)] $\tilde{\pi}_i = \log p(y_i|y)$,
  \item[(b)] $\tilde{\pi}_i = \E_y[\log p(y_i|y)]$,
  \item[(c)] $\tilde{\pi}_i = \E_{\theta \sim q}[\log p(y_i|\theta)]$,
  \item[(d)] $\tilde{\pi}_i = \log p(y_i|\E_{\theta \sim q}[\theta])$,
  \item[(e)] $\tilde{\pi}_i = \log p(y_i|\hat{\theta}_q)$.
  \item[(f)] $\tilde{\pi}_i = \log p(y_i|y)+V_{\theta \sim p(\cdot|y)} (\log p(y_i|\theta))$.
  \item[(g)] $\tilde{\pi}_i = \log p(y_i|y) - \nabla \log p(y_i|\hat{\theta})^T \Sigma_\theta \nabla \log p(y_i|\hat{\theta})$ for any given fixed $\hat{\theta}$ and where the covariance matrix is with respect to $\theta \sim p(\cdot|y)$.
  \item[(h)] $\tilde{\pi}_i = \log p(y_i|y) - \nabla \log p(y_i|\hat{\theta})^T \Sigma_\theta \nabla \log p(y_i|\hat{\theta}) - \frac12 \rm{tr}(H_{\hat{\theta}} \Sigma_\theta H_{\hat{\theta}}) \Sigma_\theta)$ for any given fixed $\hat{\theta}$ and where the covariance matrix is as in (g)
  \item[(i)] $\tilde{\pi}_i = \log p(y_i|\hat{\theta}_q) - \nabla \log p(y_i|\hat{\theta})^T \Sigma_\theta \nabla \log p(y_i|\hat{\theta})$ for any given fixed $\hat{\theta}$ and where the covariance matrix is as in (g)
  \item[(j)] $\tilde{\pi}_i = \log p(y_i|y) - \nabla \log p(y_i|\hat{\theta})^T \Sigma_\theta \nabla \log p(y_i|\hat{\theta}) - \frac12 \rm{tr}(H_{\hat{\theta}} \Sigma_\theta H_{\hat{\theta}}) \Sigma_\theta)$ for any given fixed $\hat{\theta}$ and where the covariance matrix is as in (g)    
  \item[(k)] $\tilde{\pi}_i = \log \hat{p}(y_i|y_{-i})$ as defined in (\ref{ephat}) for any weight function $r$ such that $r(\theta)>0$ for all $\theta \in \Theta$.
\end{itemize}

\end{proposition}

{\bf Note.} Part (k) holds in particular for the weight functions (\ref{weights}) and (\ref{trunc_weights}).

\medskip

{\em Remark.} By the variational BvM Theorems of \citet{wang2018frequentist}, $q$ can be taken to be either $q_{Lap}$, $q_{MF}$ or $q_{FR}$, i.e.\ the approximate posteriors of the Laplace, mean-field or full-rank variational families respectively in Proposition \ref{pa}, provided that one adopts the mild conditions in their paper.

The proof of Proposition \ref{pa} will be focused on proving (a) and then (b)-(e) will follow easily and (f)-(l) with only a few simple observations on the posterior variance of $\theta$. Note that parts (a)-(e) are contained in \citet{magnusson2019bayesian} and the proof of them is identical to that. Proposition \ref{pa} follows immediately from the following lemma.

\begin{lemma} \label{la}

With all quantities as defined above,  
\begin{equation} \label{ea}
\E_{y \sim p_t}|\pi_i - \log p(y_i|\theta_0)| \ra 0,
\end{equation}
with any of the definitions (a)-(e) of $\pi_i$ of Proposition \ref{pa}.
Furthermore,
\begin{equation} \label{eb}
\E_{y \sim p_t}|\log p(y_i|y_{-i}) - \log p(y_i|\theta_0)| \ra 0,
\end{equation}
as $n \ra \iy$.

\end{lemma}

\begin{proof}

To avoid burdening the notation unnecessarily, we write throughout the proof $\E_y$ for $\E_{y \sim p_t}$. For now, we also write $\E_\theta$ as shorthand for $\E_{\theta \sim p(\cdot|y_{-i})}$. Recall that $x_{+} = \max(x,0) = ReLU(x)$.

Hence
\begin{align*}
\E_{y}\left[\left(\log\frac{p(y_i|y_{-i})}{p(y_i|\theta_0)}\right)_{\foo}\right]
&= \E_{y}\left[\left(\log\frac{\E_{\theta}[p(y_i|\theta)]}{p(y_i|\theta_0)}\right)_{\foo}\right] \\
&\leq \E_{y}\left[\log \left(1+\frac{\E_{\theta}\left[C(y_i)p(y_i|\theta_0)\n\theta-\theta_0\n\right]}{p(y_i|\theta_0)}\right)\right] \\
&\leq \E_{y,\theta}[C(y_i) \n\theta-\theta_0\n] \\
&\leq \left(\E_{y_i}[C(y_i)^2]\E_{y,\theta}\left[\n\theta-\theta_0\n^2\right]\right)^{1/2} \\
&\ra 0 \mbox{ as } n \ra \iy.
\end{align*}

Here the first inequality follows from condition (i) and the second inequality from the fact that $\log(1+x)<x$ for $x \geq 0$. The third inequality is Schwarz inequality. The limit conclusion follows from the consistency of the posterior $p(\cdot|y_{-i})$ and the definition of weak convergence, since $\n\theta-\theta_0\n^2$ is a continuous bounded function of $\theta$ (recall that $\Theta$ is bounded) and that the first factor is finite by condition (i).

For the reverse inequality,
\begin{align*}
\E_{y}\left[\left(\log\frac{p(y_i|\theta_0)}{p(y_i|y_{-i})}\right)_{\foo}\right] 
&= \E_{y}\left[\left(\log\E_{\theta}\left[\frac{p(y_i|\theta_0)]}{p(y_i|\theta)}\right]\right)_{\foo}\right] \\
&\leq \E_{y}\left[\log \left(1+\E_{\theta}\left[\frac{C(y_i)p(y_i|\theta)\n\theta-\theta_0\n}{p(y_i|\theta)}\right]\right)\right] \\
&\leq \left(\E_{y_i}[C(y_i)^2] \E_{y,\theta}\left[\n\theta-\theta_0\n^2\right]\right)^{1/2} \\
&\ra 0 \mbox{ as } n \ra \iy.
\end{align*}

This proves (\ref{eb}) and an identical argument (now letting $\E_\theta$ stand for $\E_{\theta \sim p(\cdot|y)}$) proves (\ref{ea}) for $\tilde{\pi}_i=p(y_i|y)$.

For $\tilde{\pi}_i=-\E_y[\log p(y_i|y)]$, note first that
\begin{align*}
\E_y\left|\E_y[\log p(y_i|y)]-\E_y[\log p(y_i|y_{-i})]\right| 
&= \left|\E_y[\log p(y_i|y) - \log p(y_i|y_{-i})]\right| \\
&\leq \E_y\left|\log p(y_i|y) - \log p(y_i|y_{-i})]\right|
\end{align*}

which goes to $0$ by (\ref{eb}) and (a). Hence we can replace $\tilde{\pi}_i=-\E[\log p(y_i|y)]$ with $\tilde{\pi}_i=-\E[\log p(y_i|y_{-i})]$ when proving (b). To that end, observe that
\begin{align*}
\left(\E_y[\log p(y_i|y_{-i})]-\log p(y_i|\theta_0)\right)_+ 
&= \left(\E_{y_i}\left[\E_{y_{-i}}\left[\log\frac{p(y_i|y_{-i})}{p(y_i|\theta_0)}\right]\right]\right)_{\foo} \\
&\leq \E_y \left[ \left( \log\frac{p(y_i|y_{-i})}{p(y_i|\theta_0)}\right)_{\foo} \right].
\end{align*}
where the inequality is Jensen's inequality used twice on the convex function $x \ra x_+$. Now everything is identical to the proof of (\ref{eb}) and the reverse inequality is analogous.

The other choices of $\tilde{\pi}_i$ follow along very similar lines. For $\tilde{\pi}_i=-\log p(y_i|\hat{\theta}_q)$, we have on mimicking the above that
\begin{align*}
\E_{y}\left[\left(\log\frac{p(y_i|\hat{\theta}_q)}{p(y_i|\theta_0)}\right)_{\foo}\right] 
&\leq \left(\E_{y_i}[C(y_i)^2]\E_{y}\left[\n\hat{\theta}_q-\theta_0\n^2\right]\right)^{1/2}
\end{align*}
and $\E_{y}[\n\hat{\theta}_q-\theta_0\n^2] \ra 0$ as $n \ra \iy$ by the assumed consistency of $\hat{\theta}_q$. The reverse inequality is analogous and (\ref{ea}) for $\pi_i=p(y_i|\hat{\theta}_q)$ is established.

For the case $\tilde{\pi}_i=-\log p(y_i|\E_{\theta \sim q}\theta)$, the analogous analysis gives
\[
\E_{y}\left[\left(\log\frac{p(y_i|\E_{\theta \sim q}\theta)}{p(y_i|\theta_0)}\right)_{\foo}\right] \\
\leq \E_{y_i}[C(y_i)^2]\E_y[\n\E_{\theta \sim q}\theta - \theta_0\n^2].
\]
Since $x \ra \n x-\theta_0\n^2$ is convex, the second factor on the right hand side is bounded by $\E_{y,\theta \sim q}[\n \theta-\theta_0\n^2]$ which goes to 0 by the consistency of $q$ and the boundedness of $\Theta$. The reverse inequality is again analogous.

For $\tilde{\pi}_i = -\E_{\theta \sim q}[\log p(y_i|\theta)]$,
\begin{align*}
\E_{y}\left[\left( \E_{\theta \sim q}[\log p(y_i|\theta)] - \log p(y_i|\theta_0) \right)_{\foo} \right] 
&= \E_y\left[ \left( \E_{\theta \sim q}\left[\log\frac{p(y_i|\theta)}{p(y_i|\theta_0)}\right] \right)_{\foo} \right] \\
&\leq \E_{y,\theta \sim q}\left[ \left(\log\frac{p(y_i|\theta)}{p(y_i|\theta_0)}\right)_{\foo} \right] \\
&\leq \left(\E_{y_i}[C(y_i)^2]\E_{y,\theta \sim q}[\n\theta-\theta_0\n^2]\right)^{1/2} \ra 0
\end{align*}
as $n \ra \iy$ by the consistency of $q$. Here the first inequality is Jensen's inequality applied to $x \ra x_+$ and the second inequality follows along the same lines as before.

To prove (f) it suffices by the triangle inequality to prove that $\E_y[V_{\theta \sim p(\cdot|y)}(\log p(y_i|\theta))] \ra 0$ as $n \ra \iy$. This follows from

\begin{align*}
\E_y\left[\E_{\theta \sim p(\cdot|y)} \left[(\log p(y_i|\theta)-\log p(y_i|\theta_0))_+^2\right]\right] 
&\leq \E_{y,\theta} \left[\log \left( 1 + \frac{C(y_i)p(y_i|\theta)\n\theta-\theta_0\n}{p(y_i|\theta_0)}\right)^2\right] \\
&\leq \E_{y,\theta}[2C(y_i) \n \theta-\theta_0 \n] \\
&\leq 2\E_{y,\theta}[C(y_i)^2]^{1/2}\E_{y,\theta}[\n\theta-\theta_0\n^2]^{1/2} \ra 0.
\end{align*}

To prove that $\E_y[\E_{\theta \sim p(\cdot|y)} \left[(\log p(y_i|\theta_0)-\log p(y_i|\theta))_+^2\right] \ra 0$ is analogous.

\medskip

For (g) and (h) it suffices to observe that $\max_{i,j}|\Cov(\theta(i),\theta(j))| \ra 0$. However
\begin{align*}
|\max_{i,j}\Cov(\theta(i),\theta(j))| &= \max_iV(\theta(j)) \\
&\leq \max_i\E[|\theta(i)-\theta_0(i)|^2] \\
&\ra 0
\end{align*}
where the final conclusion follows from the consistency of $\theta \sim p(\cdot|y)$ and the boundedness of $\Theta$. Hence (g) and (h) are established. Similarly (g2) and (h2) follows from (g), (h) and (e).

For (k), write $r'(\theta_s) = r(\theta_s)/\sum_{j=1}^{S}r(\theta_j)$ for the random weights given to the individual $\theta_s$:s in the expression for $\hat{p}(y_i|y_{-i})$. Then we have, with $\theta=(\theta_1,\ldots,\theta_S)$ chosen according to $q$,
\begin{align*}
\E_y \left[ \left(\log \frac{\hat{p}(y_i|y_{-i})}{p(y_i|\theta_0)}\right)_{\foo}\right] 
&= \E_{y,\theta}\left[\left(\log\frac{\sum_{s=1}^{S}r'(\theta_s)p(y_i|\theta_s)}{p(y_i|\theta_0)}\right)_{\foo}\right] \\
&\leq \E_{y,\theta}\left[\log\left(1+\frac{\sum_{s=1}^{S}r'(\theta_s)|p(y_i|\theta_s)-p(y_i|\theta_0)|}{p(y_i|\theta_0)}\right)\right] \\
&\leq \E_{y,\theta}\left[\log\left(1+C(y_i)\sum_{s=1}^{S}r'(\theta_s)\n\theta_s-\theta_0\n\right)\right] \\
&\leq \E_{y,\theta}\left[\log\left(1+C(y_i)\sum_{s=1}^{S}\n\theta_s-\theta_0\n\right)\right] \\
&\leq \E_{y,\theta}\left[C(y_i)\sum_{s=1}^{S} \n\theta_s-\theta_0\n\right] \\
&\leq \left(\E_{y_i}[C(y_i)^2]\E_{y,\theta}\left[\left(\sum_{s=1}^{S} \n\theta_s-\theta_0\n\right)^2\right]\right)^{1/2},
\end{align*}

where the second inequality is condition (i) and the limit conclusion follows from the consistency of $q$.
For the reverse inequality to go through analogously, observe that
\begin{align*}
\frac{\left| p(y_i|\theta_0)-\sum_s r'(\theta_s)p(y_i|\theta_s) \right|}{\sum_s r'(\theta_s)p(y_i|\theta_s)} 
&\leq \frac{\sum_s r'(\theta_s)|p(y_i|\theta_s)-p(y_i|\theta_0)|}{\sum_s r'(\theta_s)p(y_i|\theta_s)} \\
&\leq \frac{\sum_s r'(\theta_s)p(y_i|\theta_s)\n \theta_s-\theta_0 \n}{\sum_s r'(\theta_s)p(y_i|\theta_s)} \\
&\leq \max_s\n\theta_s-\theta_0\n \\
&\leq \sum_s \n \theta_s-\theta_0 \n.
\end{align*}
Equipped with this observation, mimic the above.


\end{proof}

\section*{Reproducing results}

\subsection*{The arsenic data}

For the spline model comparison we use the \texttt{rstanarm} R package \citep{rstanarm} with the following R script.

\begin{lstlisting}[language=R,basicstyle=\footnotesize\ttfamily]
#' **Load data**
url <- 
  "http://stat.columbia.edu/~gelman/arm/examples/arsenic/wells.dat"
wells <- read.table(url)
wells$dist100 <- with(wells, dist / 100)
wells$y <- wells$switch

#' **Centering the input variables**
wells$c_dist100 <- wells$dist100 - mean(wells$dist100)
wells$c_arsenic <- wells$arsenic - mean(wells$arsenic)
wells$c_educ4 <- wells$educ/4 - mean(wells$educ/4)

#* **Latent linear model no interactions**
fit_1 <- stan_glm(y ~ c_dist100 + c_arsenic + c_educ4,
                  family = binomial(link="logit"), 
                  data = wells, 
                  iter = 1500, 
                  warmup = 1000, 
                  chains = 4)

#* **Latent linear model**
fit_2 <- stan_glm(y ~ c_dist100 + c_arsenic + c_educ4 +
                      c_dist100:c_educ4 + c_arsenic:c_educ4,
                  family = binomial(link="logit"), 
                  data = wells, 
                  iter = 1500, 
                  warmup = 1000, 
                  chains = 4)

#* **Latent GAM**
fit_3 <- stan_gamm4(y ~ s(dist100) + s(arsenic) + s(dist100, c_educ4),
                    family = binomial(link="logit"), 
                    data = wells, 
                    iter = 1500, 
                    warmup = 1000, 
                    chains = 4)
\end{lstlisting}           



\subsection*{Generating data and fitting regularized horse-shoe and normal model}
\begin{lstlisting}[language=R,basicstyle=\footnotesize\ttfamily]
library(arm)
library(rstanarm)

n <- 1e6

set.seed(1656)
x <- rnorm(n)
xn <- matrix(rnorm(n*99),nrow=n)
a <- 2
b <- 3
sigma <- 10
y <- a + b*x + sigma*rnorm(n)
fake <- data.frame(x, xn, y)

fit1 <- stan_glm(y ~ ., data=fake, 
                 mean_PPD=FALSE, 
                 refresh=0, 
                 seed=SEED, 
                 chains = 4, 
                 warmup = 1000, 
                 iter = 1500)
                 
fit2 <- stan_glm(y ~ ., prior=hs(), data=fake,
                 mean_PPD=FALSE, 
                 refresh=0, 
                 seed=SEED, 
                 chains = 4, 
                 warmup = 1000, 
                 iter = 1500)              
\end{lstlisting}

\section*{Models}

\subsection*{Stan Models}

\subsubsection*{Bayesian linear regression (BLR)}
\begin{lstlisting}[language=Stan,basicstyle=\footnotesize\ttfamily]
data {
  int <lower=0> N; 
  int <lower=0> D; 
  matrix [N, D] X; 
  vector [N] y;
}

parameters {
  vector [D] beta;
  real <lower=0> sigma;
}

model {
  // prior
  target += normal_lpdf(beta | 0, 10);
  target += normal_lpdf(sigma | 0, 1);
  // likelihood
  target += normal_lpdf(y | X * beta, sigma);
}
\end{lstlisting}

\subsubsection*{Pooled model (1)}
\begin{lstlisting}[language=Stan,basicstyle=\footnotesize\ttfamily]
data {
  int<lower=0> N; 
  vector[N] floor_measure;
  vector[N] log_radon;
}

parameters {
  real alpha;
  real beta;
  real<lower=0> sigma_y;
} 

model {
  vector[N] mu;
  
  // priors
  sigma_y ~ normal(0, 1);
  alpha ~ normal(0, 10);
  beta ~ normal(0, 10);

  // likelihood  
  mu = alpha + beta * floor_measure;
  for(n in 1:N){
      target += normal_lpdf(log_radon[n]| mu[n], sigma_y);
  }
}

\end{lstlisting}

\subsubsection*{Partially pooled model (2)}
\begin{lstlisting}[language=Stan,basicstyle=\footnotesize\ttfamily]
data {
  int<lower=0> N; 
  int<lower=0> J; 
  int<lower=1,upper=J> county_idx[N];
  vector[N] log_radon;
} 
parameters {
  vector[J] alpha_raw;
  real mu_alpha;
  real<lower=0> sigma_alpha;
  real<lower=0> sigma_y;
} 

transformed parameters {
  vector[J] alpha;
  // implies: alpha ~ normal(mu_alpha, sigma_alpha);
  alpha = mu_alpha + sigma_alpha * alpha_raw;
}

model {
  vector[N] mu;
  
  // priors
  sigma_y ~ normal(0,1);
  sigma_alpha ~ normal(0,1);
  mu_alpha ~ normal(0,10);
  alpha_raw ~ normal(0, 1);

  // likelihood 
  for(n in 1:N){
    mu[n] = alpha[county_idx[n]];
    target += normal_lpdf(log_radon[n] | mu[n], sigma_y);
  }
}
\end{lstlisting}

\subsubsection*{No pooled model (3)}
\begin{lstlisting}[language=Stan,basicstyle=\footnotesize\ttfamily]
data {
  int<lower=0> N; 
  int<lower=0> J; 
  int<lower=1,upper=J> county_idx[N];
  vector[N] floor_measure;
  vector[N] log_radon;
} 

parameters {
  vector[J] alpha;
  real beta;
  real<lower=0> sigma_y;
} 

model {
  vector[N] mu;
  // Prior
  sigma_y ~ normal(0, 1);
  alpha ~ normal(0, 10);
  beta ~ normal(0, 10);

  // Likelihood
  for(n in 1:N){
    mu[n] = alpha[county_idx[n]] + beta * floor_measure[n];
    target += normal_lpdf(log_radon[n] | mu[n], sigma_y);
  }
}

\end{lstlisting}

\subsubsection*{Variable intercept model (4)}
\begin{lstlisting}[language=Stan,basicstyle=\footnotesize\ttfamily]
data {
  int<lower=0> J; 
  int<lower=0> N; 
  int<lower=1,upper=J> county_idx[N];
  vector[N] floor_measure;
  vector[N] log_radon;
} 
parameters {
  vector[J] alpha_raw;
  real beta;
  real mu_alpha;
  real<lower=0> sigma_alpha;
  real<lower=0> sigma_y;
} 

transformed parameters {
  vector[J] alpha;
  // implies: alpha ~ normal(mu_alpha, sigma_alpha);
  alpha = mu_alpha + sigma_alpha * alpha_raw;
}

model {
  vector[N] mu;
  
  // Prior
  sigma_y ~ normal(0,1);
  sigma_alpha ~ normal(0,1);
  mu_alpha ~ normal(0,10);
  beta ~ normal(0,10);
  alpha_raw ~ normal(0, 1);
  
  for(n in 1:N){
    mu[n] = alpha[county_idx[n]] + floor_measure[n]*beta;
    target += normal_lpdf(log_radon[n]|mu[n],sigma_y);
  }
}
\end{lstlisting}

\subsubsection*{Variable slope model (5)}
\begin{lstlisting}[language=Stan,basicstyle=\footnotesize\ttfamily]
data {
  int<lower=0> J; 
  int<lower=0> N; 
  int<lower=1,upper=J> county_idx[N];
  vector[N] floor_measure;
  vector[N] log_radon;
}

parameters {
  real alpha;
  vector[J] beta_raw;
  real mu_beta;
  real<lower=0> sigma_beta;
  real<lower=0> sigma_y;
} 

transformed parameters {
  vector[J] beta;
  // implies: beta ~ normal(mu_beta, sigma_beta);
  beta = mu_beta + sigma_beta * beta_raw;  
}

model {
  vector[N] mu;
  // Prior
  alpha ~ normal(0,10);
  sigma_y ~ normal(0,1);
  sigma_beta ~ normal(0,1);
  mu_beta ~ normal(0,10);
  beta_raw ~ normal(0, 1);

  for(n in 1:N){
    mu[n] = alpha + floor_measure[n] * beta[county_idx[n]];
    target += normal_lpdf(log_radon[n]|mu[n],sigma_y);
  }
}
\end{lstlisting}

\subsubsection*{Variable intercept and slope model (6)}
\begin{lstlisting}[language=Stan,basicstyle=\footnotesize\ttfamily]
data {
  int<lower=0> N;
  int<lower=0> J;
  int<lower=1,upper=J> county_idx[N];
  vector[N] floor_measure;
  vector[N] log_radon;
}

parameters {
  real<lower=0> sigma_y;
  real<lower=0> sigma_alpha;
  real<lower=0> sigma_beta;
  vector[J] alpha_raw;
  vector[J] beta_raw;
  real mu_alpha;
  real mu_beta;
}
transformed parameters {
  vector[J] alpha;
  vector[J] beta;
  // implies: alpha ~ normal(mu_alpha, sigma_alpha);
  alpha = mu_alpha + sigma_alpha * alpha_raw;
  // implies: beta ~ normal(mu_beta, sigma_beta);
  beta = mu_beta + sigma_beta * beta_raw;  
}

model {
  vector[N] mu;
  // Prior
  sigma_y ~ normal(0,1);
  sigma_beta ~ normal(0,1);
  sigma_alpha ~ normal(0,1);
  mu_alpha ~ normal(0,10);
  mu_beta ~ normal(0,10);
  alpha_raw ~ normal(0, 1);
  beta_raw ~ normal(0, 1);
  
  // Likelihood
  for(n in 1:N){
    mu[n] = alpha[county_idx[n]] + floor_measure[n] * beta[county_idx[n]];
    target += normal_lpdf(log_radon[n] | mu[n], sigma_y);
  }
}
\end{lstlisting}

\newpage


\end{document}